 \documentclass[journal]{IEEEtran}
\usepackage{algorithmic}

\usepackage[T1]{fontenc}% optional T1 font encoding

\usepackage{graphicx}
\usepackage[noadjust]{cite}
\usepackage{mcite}
\usepackage{amsfonts,helvet}
\usepackage{fancyhdr}
\usepackage{threeparttable}
\usepackage{epsf,epsfig}
\usepackage{amsthm}
\usepackage{amsmath}
\usepackage{siunitx}
\usepackage{amssymb}
\usepackage{dsfont}
\usepackage{subfigure}
\usepackage{color}
\usepackage{enumerate}
\usepackage{hyperref}
\usepackage{cancel}
\usepackage{bbm}
\usepackage{dsfont}
\usepackage[subnum]{cases}
\usepackage{adjustbox}
\usepackage[linesnumbered,ruled,noend]{algorithm2e}
\usepackage{multicol}
\usepackage[english]{babel}
\usepackage{enumitem}

\newtheorem{theorem}{Theorem}

\newtheorem{corollary}{Corollary}

\newmuskip\pFqmuskip

\def\be{{\bf e}}
\def\bff{{\bf f}}

\def\bh{{\bf h}}

\def\bq{{\bf q}}

\def\bs{{\bf s}}

\def\bw{{\bf w}}
\def\bx{{\bf x}}

\def\bC{{\bf C}}
\def\bD{{\bf D}}
\def\bE{{\bf E}}

\def\bH{{\bf H}}
\def\bI{{\bf I}}

\def\bK{{\bf K}}

\def\bM{{\bf M}}

\def\bW{{\bf W}}

\def\bZ{{\bf Z}}

\def\cC{\mbox{$\mathcal{C}$}}

\def\cL{\mbox{$\mathcal{L}$}}

\def\cN{\mbox{$\mathcal{N}$}}

\def\bbC{\mbox{$\mathbb{C}$}}

\def\bbE{\mbox{$\mathbb{E}$}}

\def\bbR{\mbox{$\mathbb{R}$}}

\def\ubM{\mbox{${\bf{M}}$}}

\def\ubW{\mbox{${\bf{W}}$}}

\def\btau{\boldsymbol{\tau}}
\def\bSigma{\boldsymbol{\Sigma}}

\usepackage{array}

\makeatletter
\newcommand{\thickhline}{%
    \noalign {\ifnum 0=`}\fi \hrule height 1pt
    \futurelet \reserved@a \@xhline
}
\newcolumntype{"}{@{\hskip\tabcolsep\vrule width 1pt\hskip\tabcolsep}}

\IEEEoverridecommandlockouts

\title{
Coordinated Beamforming in Quantized Massive MIMO Systems with Per-Antenna Constraints
}
\author{
Yunseong Cho, {\it Graduate Student Member, IEEE,} Jinseok Choi, {\it Member, IEEE,}  \\
and Brian L. Evans, {\it Fellow, IEEE} \thanks{

Y. Cho and B. L. Evans are with the Wireless Networking and Communication Group (WNCG), Dept. of Electrical and Computer Engineering, The University of Texas at Austin, Austin, TX 78712 USA. (e-mail: 
yscho@utexas.edu, bevans@ece.utexas.edu).
J. Choi is with Dept. of Electrical and Computer Engineering, Ulsan National Institute of Science and Technology (UNIST), Ulsan, South Korea
(e-mail:jinseokchoi@unist.ac.kr).
}
}

\begin{document}
\maketitle
%%%%%%%%%%%%%%%%%%%%%%%%%%%%
\begin{abstract}
In this work, we present a solution for coordinated beamforming for large-scale downlink (DL) communication systems with low-resolution data converters when employing a per-antenna power constraint that limits the maximum antenna power to alleviate hardware cost. 
To this end, we formulate and solve the antenna power minimax problem for the coarsely quantized DL system with target signal-to-interference-plus-noise ratio requirements. 
We show that the associated Lagrangian dual with uncertain noise covariance matrices achieves zero duality gap and that the dual solution can be used to obtain the primal DL solution.
Using strong duality, we propose an iterative algorithm to determine the optimal dual solution, which is used to compute the optimal DL beamformer.
We further update the noise covariance matrices using the optimal DL solution with an associated subgradient and perform projection onto the feasible domain.
Through simulation, we evaluate the proposed method in maximum antenna power consumption and peak-to-average power ratio which are directly related to hardware efficiency.
% In this work, we present a coordinated beamforming solution with per-antenna constraints to minimize transmit power of internet-of-things (IoT) communication systems where multiple access points (APs) equipped with a number of antennas and low-resolution data converters serve dedicated IoT devices.
% Considering a realistic deployment, we solve the antenna power minimax problem of the coarsely quantized downlink (DL) system with target quality-of-service (QoS) constraints. 
% We first derive a Lagrangian dual problem of the DL problem.
% The associated Lagrangian dual with uncertain noise covariance matrices makes the problem more manageable, showing that there exists no duality gap.
% Motivated by the strong duality, we propose an iterative algorithm to determine the optimal dual solution, which is eventually used to compute the optimal DL beamformer.
% We further update the noise covariance using the optimal DL solution with an associated subgradient and then perform proper projection onto the feasible domain.
% Simulation results validate  the proposed method in terms of maximum antenna power consumption and peak-to-average power ratio which are directly related with hardware efficiency.
\end{abstract}
%%%%%%%%%%%%%%%%%%%%%%%%%%%%

%%%%%%%%%%%%%%%%%%%%%%%%%%%
\section{Introduction}
\label{sec:intro}
%%%%%%%%%%%%%%%%%%%%%%%%%%%
Massive multiple-input-multiple-output (MIMO) systems have drawn attention for fifth-generation wireless communication systems \cite{marzetta2010noncooperative}.
However, a large number of analog-to-digital converters (ADCs) and digital-to-analog converters (DACs) connected to the antennas require prohibitively high power consumption.
Accordingly, employing low-resolution quantizers has gathered momentum as a power-efficient solution \cite{choi2016near,studer2016quantized,choi2019robust,cho2019one,choi2017resolution,  choi2019two,choi2020quantized,nguyen2021linear,yuan2020toward,Park2021Construction}. 
The design of multicell systems should take into account intra-cell and inter-cell interferers as well as the quantization error.
Moreover, for a realistic deployment, a per-antenna power constraint that limits the power of each antenna is desirable by allowing the system to operate with more energy-efficient power amplifiers \cite{yu2007transmitter, dahrouj2010coordinated}.

%In this regard, we investigate coordinated multipoint (CoMP) beamforming (BF) and power control (PC) problems in low-resolution multicell massive MIMO systems incorporated with per-antenna power constraints.

% ChEst & Detection
In recent years, low-resolution converters have been introduced in many communication systems and their corresponding algorithms \cite{choi2016near, studer2016quantized, choi2019robust,cho2019one,choi2020quantized,choi2019two,nguyen2021linear,yuan2020toward,Park2021Construction}.
The authors in \cite{choi2019robust} employed artificial noise to precisely learn the likelihood probabilities for maximum likelihood detection under one-bit ADCs.
In \cite{cho2019one}, computation of soft metrics with one-bit observations and its application to channel coding were presented.
% resolution-adaptive-ADC
A hybrid MIMO architecture with resolution-adaptive ADCs for
millimeter wave communication was developed in \cite{choi2017resolution}.
Linear estimation and model-based deep neural networks approach were combined in \cite{nguyen2021linear} for data detection with one-bit ADCs.
Considering different traffic load requirements of IoT devices, the authors in \cite{yuan2020toward} developed a grant-free access scheme for massive MIMO systems with mixed-ADC access points (APs).

% AQNM
%For analytic tractability, an additive quantization noise model (AQNM) was used in the recent literatures \cite{choi2019two, choi2017resolution, choi2020quantized} 
% Low-resolution DAC %

\begin{figure}[!t]\centering
	\includegraphics[width=0.95\columnwidth]{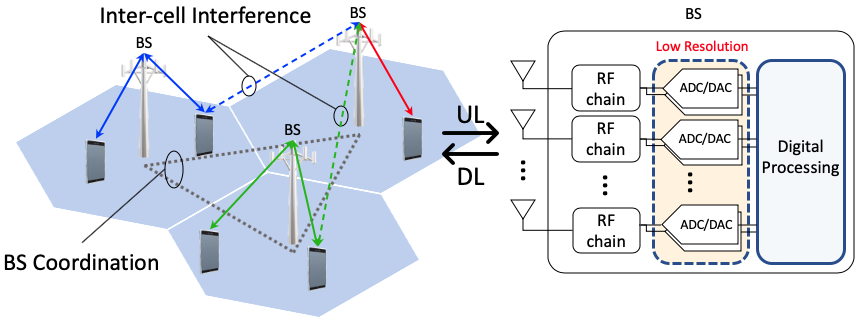}
	\vspace{-1.0em}
	\caption{Multicell multiuser-MIMO configuration with low-resolution ADCs and DACs at the base station (BS).} 
	\label{fig:system}
	\vspace{-1.0em}
\end{figure}

Coordinated multipoint (CoMP) 
%% designs across base stations (BSs) 
has been explored to increase throughput, coverage and reliability \cite{rashid1998joint,rashid1998transmit,dahrouj2010coordinated,choi2020quantized}.
In \cite{rashid1998joint}, CoMP beamforming (BF) and power control (PC) in uplink (UL) were presented.
Considering downlink (DL) as a virtual UL, UL-DL CoMP BF and PC were further proposed in \cite{rashid1998transmit} in a distributed manner using local measurements.
The authors in \cite{dahrouj2010coordinated} showed Lagrangian-based duality for multiuser MIMO systems and proposed a distributed method to lower computational load on users and BSs.
CoMP BF and PC were also studied for coarsely quanitized massive MIMO systems, and their closed-form solution and extension to a wideband scenario were investigated \cite{choi2020quantized}. 
The synergy between massive MIMO and CoMP is studied with higher throughput \cite{jungnickel2014role}.
%

% Contribution
In this paper, we design CoMP-based BF with the per-antenna power constraint in coarsely quantized large-scale MIMO systems.
%%% while supporting massive IoT connectivity.
We first formulate the DL antenna power minimax problem with individual signal-to-interference-plus-noise-ratio (SINR) constraints.
We then derive the Lagrangian dual of the DL problem which can be considered as a virtual UL problem with uncertain noise covariance matrices. 
By transforming the DL problem to a strictly feasible second-order cone program (SOCP), we derive strong duality between the DL problem and its dual.
Leveraging strong duality, we propose an iterative algorithm to solve the dual in a distributed manner for fixed noise covariance matrices. 
The solutions of the dual are used to obtain an optimal DL BF.
The optimal DL solutions update the noise covariance matrices which are then projected onto the feasible constraint set.
%Numerical results show that the proposed design outperforms a conventional method in terms of total power and achievable SINR.

%%% NOTATIONS %%%
{\it Notation}: $\bf{A}$ is a matrix and $\bf{a}$ is a column vector. 
$\mathbf{A}^{H}$ and $\mathbf{A}^T$  denote conjugate transpose and transpose. 
$[{\bf A}]_{i,:}$ and $ \mathbf{a}_i$ indicate the $i$th row and column vectors of $\bf A$. 
We denote $a_{i,j}$ as the $\{i,j\}$th element of $\bf A$ and $a_{i}$ as the $i$th element of $\bf a$. 
$\mathcal{CN}(\mu, \sigma^2)$ is a complex Gaussian distribution with mean $\mu$ and variance $\sigma^2$. 
The diagonal matrix $\rm diag(\bf A)$ has $\{a_{i,i}\}$ as its diagonal entries, and $\rm diag (\bf a)$ or $\rm diag({\bf a}^T)$ has $\{a_i\}$ as its diagonal entries. 
A block diagonal matrix is presented as ${\rm blkdiag}({\bf A}_1, \dots,{\bf A}_{N})$. 
${\bf I}_N$ is a $N\times N$ identity matrix and ${\bf 1}_N$ is a $N \times 1$ ones vector.
$\otimes$ represents Kronecker product.
% Inequality between vectors, e.g. $\bf a < b$, is element-wise inequality.

%%%%%%%%%%%%%%%%%%%%%%%%%%%%
\section{System Model}
\label{sec:sys_model}
%%%%%%%%%%%%%%%%%%%%%%%%%%%%

We consider a multicell multiuser-MIMO network with $N_c$ cells, $N_{u}$ single-antenna users per cell.
BS$_i$ represents the BS with $N_b$ antennas in cell $i$ that serves $N_u$ dedicated users.
We assume that the BSs distributed over $N_c$ cells cooperate and employ low-resolution DACs.
%Time division multiplexing (TDD) is assumed.

%%%%%%%%%%%%%%%%%%%%%%%%%%%%%%%%%%
\subsection{DL System Models with Low-Resolution DACs}
%%%%%%%%%%%%%%%%%%%%%%%%%%%%%%%%%%

Let $\bs_i \!\in\! \bbC^{N_u}$ denote the symbols for $N_u$ users in cell $i$ and $\bW_i\!=\!\left[\bw_{i,1},\ldots,\bw_{i,N_c}\right] \!\in\! \bbC^{N_b \times N_u} $ collects the associated precoders. 
We define the precoded signals as $\bx_i = \bW_i\bs_i$.
% \begin{align}
%     \nonumber
%     \bu_i = \bW_i\bs_i,
% \end{align} 
% where ${\pmb \Lambda}_i = {\rm diag}\big(\lambda_{i,1},\dots,\lambda_{i,N_u}\big)$ is the diagonal matrix of transmit power.
Upon generating $\bx_i$, the signal is treated by the low-resolution DACs with $b$ quantization bits.
For analytic tractability, we adopt the AQNM~\cite{choi2020quantized, orhan2015low} that delivers a linear approximation of the quantization process derived from assuming a scalar minimum-mean-squared-error (MMSE) quantizer.
% The AQNM is accurate enough in low and medium SNR ranges \cite{orhan2015low}.
Under the AQNM, the quantized signal is modeled as
\begin{align}
    \bx_{{\rm q},i} = \alpha \bx_i + \bq_i
\end{align}
where $\bq_i \sim \cC\cN({\bf 0}, \bC_{{\bf q}_i{\bf q}_i})$ is the quantization noise vector at BS$_i$ with its covariance matrix of \cite{choi2020quantized}
\begin{align}
    \label{eq:Cqq_dl}
    \bC_{{\bf q}_i{\bf q}_i} = \alpha\beta{\rm diag}\big(\bW_i\bW_i^H\big).
\end{align}
The quantization gain $\alpha$ is defined as $\alpha \!=\! 1- \beta$ where $\beta$'s are listed in Table 1 in \cite{fan2015uplink} for $b \leq 5$ assuming $\bs_i\!\sim\!\mathcal{CN}({\bf 0}_{N_u},{\bf I}_{N_u}), \forall i\in\{1,\ldots,N_c\}$. 

Due to the broadcast channel, users in cell $i$ receive signals from all BSs. 
By defining $\bH_{i,j}$ as the channel between BS$_i$ and users in cell $j$, the received signal at user $u$ in cell $i$ is
\begin{align}
    y_{i,u} \!=\!\ &\alpha {\bf h}_{i,i,u}^H \!{\bf w}_{i,u} {s}_{i,u} \nonumber \\ 
    &+ \alpha\!\!\!\!\!\! \sum_{(j,v)\neq (i,u)}^{(N_c,N_u)} \!\!\!\!{\bf h}_{j,i,v}^H  \!{\bf w}_{j,v} {s}_{j,v} \!+\! \sum_{j=1}^{N_c} {\bf h}_{j,i,u}^H \bq_j \!+\!  n_{i,u} \label{eq:y_dl}
    % y^{\rm dl}_{i,u} = \alpha {\bf h}_{i,i,u}^H {\bf w}_{i,u} {s}_{i,u}^{\rm dl} + \alpha \sum_{v \neq u}^{N_u}{\bf h}^H_{i,i,u} {\bf w}_{i,v} {s}^{\rm dl}_{i,v} + \alpha \sum_{\substack{ j \neq i\\ v}}^{N_c,N_u} {\bf h}_{j,i,v}^H  {\bf w}_{j,v} {s}^{\rm dl}_{j,v} + \sum_{j=1}^{N_c} {\bf h}_{j,i,u}^H \bq_j +  n^{\rm dl}_{i,u}.
\end{align}
where the DL channel between BS$_j$ and user $u$ in cell $i$ is $\bh^H_{j,i,u}$, and $n_{i,u}\sim\mathcal{CN}(0, 1)$.
Based on \eqref{eq:Cqq_dl} and \eqref{eq:y_dl} the DL SINR for user $u$ in cell $i$ is expressed as 
\begin{align}
    \label{eq:sinr_dl_ofdm}
    &\Gamma_{i,u}=
    \frac{\alpha^2|\bh_{i,i,u}^H\bw_{i,u}|^2}{\alpha^2 \sum_{(j,v) \neq (i,u)}^{N_c, N_u}\!|\bh_{j,i,u}^H\bw_{j,v}|^2 + {\rm Q}_{i,u} +\! \sigma^2}
    %  \nonumber
    %  &\frac{\alpha^2|\bh_{i,i,u}^H\bw_{i,u}|^2}{\alpha^2 \sum_{(j,v) \neq (i,u)}^{N_c, N_u}\!|\bh_{j,i,u}^H\bw_{j,v}|^2 \!+\! \alpha\beta\!\sum_{j=1}^{N_c}\ubg_{j,i,u}^H\!{\pmb \Psi}_{N_b}{\rm diag}\big({\pmb \Psi}_{N_b}^H\ubW_j\ubW_j^H{\pmb \Psi}_{N_b}\big){\pmb \Psi}_{N_b}^H\ubg_{j,i,u}\!\! +\! \sigma^2}
\end{align}
with the quantization noise terms of
\begin{align}
	\label{eq:quantizationnoise}
	{\rm Q}_{i,u} = \sum_{j=1}^{N_c}\bh_{j,i,u}^H\bC_{\bq_j,\bq_j}\bh_{j,i,u}.
\end{align}
Finally, the DL transmit power minimization problem with per-antenna power constraints is formulated  as 
\begin{align}
    \label{eq:dl_problem}
   \mathop{{\text{minimize}}}_{{\bf w}_{i,u},\ p_{\rm 0}}& \;\;  p_{\rm 0}\\
   \label{eq:dl_sinr}
   {\text{subject to}} & \;\; \Gamma_{i,u} \geq \gamma_{i,u}\\
   \label{eq:per_antenna_const}
   &\  \Big[\bbE[\bx_{{\rm q},i}\bx_{{\rm q},i}^H]\Big]_{m,m}\leq p_{\rm 0} \\
   \nonumber
   &\;\;\forall i,u,m.
\end{align}
% Note that $\alpha$ in the objective function is a fixed positive scalar which does not change the solution of $\cP4$.
Since directly solving the above problem is difficult, we first derive a Lagrangian dual and solve the problem in the dual space with a more efficient solver.

%%%%%%%%%%%%%%%%%%%%%%%%%%%%%%%%%%%
\subsection{Lagrangian Dual Problem}
%%%%%%%%%%%%%%%%%%%%%%%%%%%%%%%%%%%

Let us collect Lagrangian multipliers as  ${\pmb \Lambda} = {\rm blkdiag}\big({\pmb \Lambda}_{1},\dots, {\pmb \Lambda}_{N_c}\big)$ where ${\pmb \Lambda}_i = {\rm diag}\big(\lambda_{i,1},\dots,\lambda_{i,N_u}\big)$.
In Theorem~\ref{thm:duality_ofdm}, a virtual UL problem is derived as a dual problem of \eqref{eq:dl_problem}.
\begin{theorem}[Duality]
    \label{thm:duality_ofdm}
    The Lagrangian dual problem of the DL problem in \eqref{eq:dl_problem} is equivalent to 
    \begin{align}
        \label{eq:dual_problem}
        \max_{{\bf D}_i }\min_{\lambda_{i,u}}& \;\;  \sum_{i,u}^{N_c,N_u}\lambda_{i,u}\sigma^2\\
        \label{eq:dual_sinr_const}
        {\text{ \rm subject to}} & \;\; \max_{\bff_{i,u}}\hat{\Gamma}_{i,u} \geq \gamma_{i,u},\\
        \label{eq:D_const1}
        &\;\;  \bD_i \succeq 0, \ \bD_i\in \bbR^{N_b\times N_b}: \text{\rm diagonal}, \\
        \label{eq:D_const2}
        &\;\;  {\rm tr}(\bD_i) \leq N_b \quad \forall i,u 
   \end{align}
   with UL SINR of
   \begin{align}
        \label{eq:sinr_ul_ofdm}
        \hat{\Gamma}_{i,u}\! =\! \frac{ \alpha^2\lambda_{i,u}|\bff_{i,u}^H\bh_{i,i,u}|^2 }{ \bff_{i,u}^H\bZ_{i,u}\bff_{i,u}} 
    \end{align}
    where
    \begin{align} 
    	\nonumber
    	\bZ_{i,u}= &\alpha^2\!\!\!\!\!\!\! \sum_{(j,v)\neq (i,u)} 
    	\!\!\!\!\!\!\!\lambda_{j,v} \bh_{i,j,v}\bh_{i,j,v}^H \\
    	&+ \alpha^2 {\bD_i}+  \alpha\beta{\rm diag}(\bH_i{\pmb \Lambda} \bH_i^H\!+\! \bD_i).
    \end{align}
\begin{proof}
    Considering $\bff_{i,u}$ as a combiner for user $u$ in cell $i$, we let $\bff_{i,u}$ be the MMSE equalizer as
    \begin{align}
    	\label{eq:mmse}
    	\bff_{i,u} =  &\bZ_{i,u}^{-1}\bh_{i,i,u}.
%    	\bff_{i,u} =  &\left(\alpha^2\!\!\! \sum_{(j,v)\neq (i,u)} \!\!\lambda_{j,v} \bh_{i,j,v}\bh_{i,j,v}^H + \alpha^2 {\bD_i}+ \hat\bC_i\right)^{-1}\bh_{i,i,u}.
    \end{align}
%    \begin{align}
%        \nonumber
%        {\bf C}_{{\bf z}_{i,u}{\bf z}_{i,u}} =\  &\alpha^2\!\!\! \sum_{(j,v)\neq (i,u)} \!\!\lambda_{j,v} \bh_{i,j,v}\bh_{i,j,v}^H + \alpha^2 {\bf I}_{N_b} \\
%        \label{eq:Czz_ul_ofdm}
%        &+ \alpha(1-\alpha) {\pmb \Psi}_{N_b}{\rm diag}\Big(\sum_{j=1}^{N_c}{\bH}_{i,j}{\pmb \Psi}_{N_u}^H{\pmb \Lambda}_j{\pmb \Psi}_{N_u}{\bH}_{i,j}+\bI_{KN_b}\Big){\pmb \Psi}_{N_b}^H.
%    \end{align}
%    Noting that ${\pmb \Psi}_{N_b}^H\ubG_{i,j} = \ubH_{i,j}{\pmb \Psi}_{N_u} ^H$, we first rewrite the diagonal matrix in \eqref{eq:Czz_ul_ofdm} as
%    \begin{align}
%        \label{eq:diag_ul_ofdm}
%        {\rm diag}\Big(\sum_{j=1}^{N_c}{\bH}_{i,j}{\pmb \Psi}_{N_u}^H{\pmb \Lambda}_j{\pmb \Psi}_{N_u}{\bH}_{i,j}+\bI_{KN_b}\Big) = {\rm diag}\Big({\pmb \Psi}_{N_b}^H\ubG_i \,{\pmb \Lambda}\,\ubG_i^H{\pmb \Psi}_{N_b}+\bI_{KN_b}\Big).
%    \end{align}
%    where $\ubG_i = [\ubG_{i,1},\dots, \ubG_{i,N_c}]$ and ${\pmb \Lambda} = {\rm blkdiag}\big({\pmb \Lambda}_{1},\dots, {\pmb \Lambda}_{N_c}\big)$.
    With the MMSE combiner that maximizes $ \hat{\Gamma}_{i,u}$, \eqref{eq:dual_sinr_const} is simplified as
    \begin{align}
        \label{eq:dual_problem_reduced}
        \max_{{\bD}_i }\min_{\lambda_{i,u}}& \;\;  \sum_{i,u}^{N_c,N_u}\lambda_{i,u}\sigma^2\\
        \label{eq:dual_sinr_const_reduced}
        {\text{ \rm subject to}} & \;\; {\bf K}_{i}({\pmb \Lambda})\preceq  \alpha \left(1 + \frac{1}{\gamma_{i,u}}\right)\lambda_{i,u}  \bh_{i,i,u}\bh_{i,i,u}^H,\\
        \nonumber
        &\;\;  \bD_i \succeq 0, \ \bD_i\in \bbR^{N_b\times N_b}: \text{\rm diagonal}, \\
        \nonumber
        &\;\;  {\rm tr}(\bD_i) \leq N_b \quad \forall i,u 
   	\end{align}
    where 
    \begin{align}
        {\bf K}_{i}({\pmb \Lambda}) \!=\! {\bf D}_{i} \! + \!\alpha\! \sum_{j,v} \!\lambda_{j,v}\bh_{i,j,v}\bh_{i,j,v}^H\! 
        +\! \beta{\rm diag}\Big(\bH_i {\pmb \Lambda}\bH_i^H\Big),
        \label{eq:K_matrix}
    \end{align}
    which is the covariance matrix of the received signal at BS$_i$.
   
    Now, we show that \eqref{eq:dual_problem_reduced} is equivalent to the Lagrangian dual of \eqref{eq:dl_problem}.
   The per-antenna constraint in \eqref{eq:per_antenna_const}  is rewritten as
    \begin{align}
    	 \Big[\bbE[\bx_{{\rm q},i}\bx_{{\rm q},i}^H]\Big]_{m,m} 
%    	 &= \Big[\bbE[\alpha^2\bx_{i}\bx_{i}^H] + \bbE[\bq_i\bq_i^H]\Big]_{m,m}\\
%    	 \nonumber
%    	 &= \Big[\alpha^2 \pmb \Psi_{N_b}^H\ubW_i\ubW_i^H\pmb\Psi_{N_b} + \alpha\beta{\rm diag}\big( \pmb \Psi_{N_b}^H\ubW_i\ubW_i^H\pmb\Psi_{N_b}\big)\Big]_{m,m}\\
%    	 & =\Big[\alpha \pmb \Psi_{N_b}^H\ubW_i\ubW_i^H\pmb\Psi_{N_b}\Big]_{m,m}\\
    	  = \left[\alpha \bW_i \bW_i^H \right]_{m,m}.
    \end{align}
    We replace the objective function in \eqref{eq:dl_problem} with $N_cN_b p_{\rm 0}$ that does not affect the problem.
    With Lagrangian multipliers $\lambda_{i,u}$ and $\mu_{i,m}$, the Lagrangian of \eqref{eq:dl_problem} is then given as 
    \begin{align}
        \nonumber
         {\cL} \!=\!  & \, N_cN_b p_{\rm 0} \!-\!\! \sum_{i,u}\!\lambda_{i,u}\Bigg(\!\frac{\alpha^2|\bw_{i,u}^H\bh_{i,i,u}|^2}{\gamma_{i,u}}
         \!-\alpha^2 \!\!\!\!\!\!\!\! \sum_{(j,v) \neq (i,u)}^{N_c, N_u}\!\!\!\!\!\!\!|\bh_{j,i,u}^H\bw_{j,v}|^2 \\
         &\!-\! {\rm Q}_{i,u} \!-\! \sigma^2\!\Bigg)
          +  \sum_{i,m}\mu_{i,m}
         \label{eq:lagrangian}
       \left(\left[\alpha \bW_i \bW_i^H \right]_{m,m} - p_{\rm 0}   \right).
%        &{\cL} =  KN_cN_b p_{\rm 0} - \sum_{i,u}\lambda_{i,u}\left(\frac{\alpha^2|\bw_{i,u}^H\bh_{i,i,u}|^2}{\gamma_{i,u}}-\alpha^2\sum_{(j,v)\neq(i,u)} |\bw_{j,v}^H\bh_{j,i,u}|^2 - \right.\\
%       & \left. \sum_{j=1}^{N_c}\ubg_{j,i,u}^H\!{\pmb \Psi}_{N_b}\bC_{\bq_j,\bq_j}{\pmb \Psi}_{N_b}^H\ubg_{j,i,u}\!\! - \sigma^2\right) + \alpha\sum_{i,m,k}\mu_{i,m,k}\left(\Big[ \pmb \Psi_{N_b}^H\ubW_i\ubW_i^H\pmb\Psi_{N_b}\Big]_{m,m} - p_{\rm 0}   \right)
    \end{align}
    First, from the proof of Theorem 2 in \cite{choi2020quantized}, we have $\sum_{i,u}\lambda_{i,u}{\rm Q}_{i,u}$ in \eqref{eq:lagrangian} rewritten as
    \begin{align}
        \label{eq:thm_proof1}
        &\sum_{i,u}\lambda_{i,u}{\rm Q}_{i,u} = \alpha\beta \sum_{i,u}\!\bw_{i,u}^H{\rm diag}\left(\bH_i
        {\pmb \Lambda}\bH_i^H\!\right)\!
        \bw_{i,u}.
    \end{align}
    Next, we let $\bD_i = {\rm diag}(\mu_{i,1}, \dots, \mu_{i,N_b})$. 
    We can cast $\sum_{i,m}\mu_{i,m}\left[\alpha \bW_i \bW_i^H \right]_{m,m}$  in \eqref{eq:lagrangian} to
    \begin{align}
    	\nonumber
    	\sum_{i,m}\mu_{i,m}\bigg[\sum_{u}\bw_{i,u}\bw_{i,u}^H\bigg]_{m,m} \!&=  \sum_{i,m}\sum_{u} w^*_{i,u,m}\mu_{i,m}w_{i,u,m}\\
    	\label{eq:thm_proof2}
    	& = \sum_{i,u}\bw^H_{i,u}\bD_i\bw_{i,u} 
   	\end{align}
    where $w_{i,u,m}$ denotes $m$th element of $\bw_{i,u}$.
    Lastly, $\sum_{i,m}\mu_{i,m}p_{\rm 0}$ can be rewritten as 
    \begin{align}
    	\label{eq:thm_proof3}
    	p_{\rm 0}\sum_{i,m}\mu_{i,m} = p_{\rm 0}\sum_{i=1}^{N_c}{\rm tr}(\bD_i).
    \end{align}

    Applying \eqref{eq:thm_proof1}, \eqref{eq:thm_proof2}, and \eqref{eq:thm_proof3} to the Lagrangian in \eqref{eq:lagrangian}, we finally have the reformed Lagrangian as
    	\begin{align}
        	{\cL} \!=\! &\sum_{i,u}\lambda_{i,u}\sigma^2 \!-\! p_{\rm 0}\!\sum_i\!\big[{\rm tr}(\bD_i)\!-\!N_b\big]\! \label{eq:lagrangian_reform}
        	\\
        	&+\!\sum_{i,u}\!\bw_{i,u}^H\!\bigg(\alpha\bD_{i}\! -\! \alpha^2\!\left(\!1\!+\!\frac{1}{\gamma_{i,u}}\!\right)\!\lambda_{i,u}\bh_{i,i,u}\bh_{i,i,u}^H \nonumber
        	\\
        	&+\alpha^2\!\sum_{j,v}\lambda_{j,v}\bh_{i,j,v}\bh_{i,j,v}^H\!+\!\alpha\beta{\rm diag}\!\left(\bH_i\,{\pmb\Lambda}\,\bH_i^H\!\right)\!\bigg)\bw_{i,u}. \nonumber
   	 	\end{align}
    Let the dual objective function be $g(\bD_i, \lambda_{i,u}) = \min_{{\bf w}_{i,u},p_{\rm 0}} \mathcal{L}$. 
   	We then need ${\rm tr}(\bD_i)\leq N_b$ and $\bK_{i} ({\pmb \Lambda}) \succeq \alpha \big(1 + 1/{\gamma_{i,u}}\big)\lambda_{i,u}  \bh_{i,i,u}\bh_{i,i,u}^H$, where $\bK_{i}({\pmb\Lambda})$ is in \eqref{eq:K_matrix}.
    Consequently, the Lagrangian dual problem of \eqref{eq:dl_problem} becomes
    \begin{align}
        \label{eq:dual_problem_reduced_inv}
        \max_{{\bD}_i }\max_{\lambda_{i,u}}& \;\;  \sum_{i,u}^{N_c,N_u}\lambda_{i,u}\sigma^2\\
        \label{eq:dual_sinr_const_reduced_inv}
        {\text{ \rm subject to}} & \;\; {\bf K}_{i}({\pmb \Lambda}) \succeq \alpha \bigg(1 + \frac{1}{\gamma_{i,u}}\bigg)\lambda_{i,u}  \bh_{i,i,u}\bh_{i,i,u}^H,\\
        \nonumber
        &\;\;  \bD_i \succeq 0, \ \bD_i\in \bbR^{N_b\times N_b}: \text{\rm diagonal}, \\
        \nonumber
        &\;\;  {\rm tr}(\bD_i) \leq N_b \quad \forall i,u 
   	\end{align}
   	% \begin{align}
    %     \nonumber
    %     \bar{\bf K}_{i}(\ubM)\! =\! \bI_{N_b} \!+\!
    %     \alpha\!\sum_{j,v}\mu_{j,v}\bh_{i,j,v}\bh_{i,j,v}^H\!+\!(1\!-\!\alpha)\bT_{N_b}{\rm diag}\!\left(\bT_{N_b}^H\ubG_i\,\ubM\,\ubG_i^H\bT_{N_b}\!\right)\!\bT_{N_b}^H
    % \end{align}
    We note that the differences between the problem in \eqref{eq:dual_problem_reduced} and in \eqref{eq:dual_problem_reduced_inv} are the reversed objectives with respect to $\lambda_{i,u}$ (i.e., $\min$ vs. $\max$) and reversed SINR inequalities in \eqref{eq:dual_sinr_const_reduced}  and \eqref{eq:dual_sinr_const_reduced_inv}.
    Since the problems in \eqref{eq:dual_problem_reduced} and \eqref{eq:dual_problem_reduced_inv} have  optimal solutions when the SINR constraints are active, the solutions for the problems are indeed equivalent to each other with the active SINR constraints.
\end{proof}
\end{theorem}
We remark that $\hat{\Gamma}_{i,u}$ in \eqref{eq:sinr_ul_ofdm} can be interpreted as the SINR of user $u$ in cell $i$ for the UL system with low-resolution ADCs, i.e., $\bff_{i,u}$ is a combiner for user $u$ in cell $i$, $\lambda_{i,u}$ is transmit power for user $u$ in cell $i$, $\bC_i$ is an aggregated quantization noise of BS$_i$ after quantization, and $\bD_i$ is a diagonal matrix of noise variances at the antennas of BS$_i$ with uncertain noise covariance in UL direction.
Accordingly, the Lagrangian dual problem is considered to be an antenna power minimax problem with noise variance constraints for a virtual UL system with low-resolution ADCs at the BSs.
\begin{corollary}[Strong Duality] Zero duality gap exists between the DL formulation and its associated dual.
    \begin{proof}
    The primal DL problem is rewritten as
    \begin{gather}
        \label{eq:strong_pf}
        \min_{{\bW},P_o} p_o \\ 
        \label{eq:strong_pf1}
        {\rm s.t.}\ \Gamma_{i,u} \geq \gamma_{i,u}, \quad \forall i,u\\ 
        \label{eq:strong_pf2}
        \left[\alpha\bW_i\bW_i^H\right]_{m,m} \leq p_o
     \end{gather}

    Let  
    $\bW_{\rm BD} \!=\! {\rm blkdiag}(\bW_1,\dots,\bW_{N_c})$,  $\tilde{\bW}_{\rm BD} \!=\! {\rm blkdiag}((\bI_{N_b}\otimes\bW_1),\dots,(\bI_{N_b}\otimes\bW_{N_c}))$, $\bE_{j,i,u} \!=\! {\rm diag}(\bh_{j,i,u}\bh_{j,i,u}^H)$, and
    $\bbE_{i,u} \!=\! {\rm vec}(\bE^{1/2}_{1,i,u},\ldots,\bE^{1/2}_{N_c,i,u})$. 
    The SINR constraints in \eqref{eq:strong_pf1} can be rewritten as
    \begin{align}
        \label{eq:strong_pf4}
        &\alpha^2 \bigg(1+\frac{1}{\gamma_{i,u}}\bigg) |{\bf w}_{i,u}^H \bh_{i,i,u}|^2 \\
%        & 
%        \left\|
%        \begin{matrix}
%            \alpha\bW_{\rm BD}^H{\rm vec}(\bh_{1,i,u},\dots,\bh_{N_c,i,u}) \\
%            \sqrt{\alpha\beta}\tilde{\bW}_{\rm BD}(0)\tilde{{\pmb \Psi}}_{N_b}(0){\rm vec}(\bE^{1/2}_{1,i,u},\!..,\bE^{1/2}_{N_c,i,u})\\
%            \vdots \\
%            \sqrt{\alpha\beta}\tilde{\bW}_{\rm BD}(K-1)\tilde{{\pmb \Psi}}_{N_b}(K-1){\rm vec}(\bE^{1/2}_{1,i,u},\!..,\bE^{1/2}_{N_c,i,u}) \\
%            \sigma
%        \end{matrix}
%        \right\|^2 \\
        &\geq  \left\|
        \begin{matrix}
            \alpha\bW_{\rm BD}^H{\rm vec}(\bh_{1,i,u},\dots,\bh_{N_c,i,u}) \\
            \sqrt{\alpha\beta} \tilde{\bW}_{\rm BD} \bbE_{i,u} \\
            \sigma
        \end{matrix}
        \right\|^2
    \end{align}
    for all $i$, $u$, and $k$. In addition, the per-antenna constraint in \eqref{eq:strong_pf2} is rewritten as
    \begin{align}
        \left[\alpha\bW_i\bW_i^H\right]_{m,m}
        % &= \frac{\alpha}{K} \sum_{\ell=0}^{K-1}\be_{m}^H\bW_i(\ell)\bW_i^H(\ell)\be_{m} \\
        = \alpha \|\be_{m}^H\bW_i \|^2
    \end{align}
	for all $m$ which is a convex constraint. 
	Accordingly, we eventually have the standard SOCP form. 
	Next, \eqref{eq:dl_problem} is strictly feasible because, for a given solution ${\bf W}$, it can be scaled by a factor of $c >1$ satisfying the constraints. 
	Thus, strong duality holds between \eqref{eq:dl_problem} and \eqref{eq:dual_problem}. 
    \end{proof}
\end{corollary}

\begin{corollary}
	\label{cor:dl_precoder_ofdm}
	% \label{cor:wb_dl_precoder_ofdm}
	% \label{cor:wb_DL_precoder}
	An optimal DL precoder forms a linear relationship with the UL MMSE receiver, i.e.,  $\bw_{i,u} = \sqrt{{\tau}_{i,u}}\bff_{i,u} \;\forall i, u$.
	Here, ${\tau}_{i,u}$ is derived from solving ${\btau} = {\bSigma}^{-1}{\bf 1}$, where  ${\bf 1}$ is a $N_uN_c \times 1$ column vector, ${\btau}=[{\btau}_{1}^T, {\btau}_{2}^T, \cdots, {\btau}_{N_c}^T ]^T$ and  ${\btau}_{i}^T = [{\tau}_{i,1}, {\tau}_{i,2}, \cdots, {\tau}_{i,N_u}]^T$, and ${\bSigma}$ is defined as
	\begin{equation}
		\label{eq:wb_constraint_matrix)}
		{\bSigma} = 
		\begin{pmatrix}
			{\bSigma}_{1,1} & {\bSigma}_{1,2} & \cdots & {\bSigma}_{1,N_c} \\
			{\bSigma}_{2,1} & {\bSigma}_{2,2} & \cdots & {\bSigma}_{2,N_c} \\
			\vdots  & \vdots  & \ddots & \vdots  \\
			{\bSigma}_{N_c,1} & {\bSigma}_{N_c,2} & \cdots & {\bSigma}_{N_c,N_c}
		\end{pmatrix},
	\end{equation}
	% Each element of $\bSigma_{i,j}\in\bbR^{N_u \times N_u}$ is defined as \eqref{eq:DLprecoder}
	and
	\begin{align} 
		% \label{eq:Wide_DLprecoder}
		&[{\bSigma}_{i,j}]_{u,v} \nonumber \\
		&=\begin{cases}
			\frac{\alpha^2}{\gamma_{i,u}}|\bh_{i,i,u}^H\bff_{i,u}|^2 \! - \alpha(1\!-\!\alpha)\bff_{i,u}^H
            {\rm diag}\big(\bh_{i,i,u}\bh_{i,i,u}^H\!\big)\bff_{i,u} 
			\\ 
			\text{if } i=j \text{, } u=v, \\
			\nonumber
			- \alpha^2 |\bh_{j,i,u}^H\bff_{j,v}|^2\! - \alpha(1\!-\!\alpha)\bff_{j,v}^H
            {\rm diag}\big(\bh_{j,i,u}\bh_{j,i,u}^H\!\big)\bff_{j,v}
			\\
			\text{otherwise.}
		\end{cases}
	\end{align}
	% where $\bh_{j,i,u}$ is the non-zero column vector that
	% belongs to $\ubg_{j,i,u}$ with the format of $\ubg_{j,i,u} = [{\bf 0}^T, \bh_{j,i,u}^T, {\bf 0}^T]^T$
	
	%  {\color{red} ($\tau_{i,u} = ?$)}
	% {\color{red} (UL combiner: $\bff_{i,u}$, DL precoder: $\bw_{i,u}$ in our paper. Avoid using $\beta$ here since we have already been using it. Try $\tau$) instead}
	\begin{proof}
		Starting from the Lagrangian in  \eqref{eq:lagrangian_reform}, we find the derivative of the Lagrangian regarding $\bw_{i,u}$ as
		\begin{align}
			&2\bigg(\alpha\bD_{i}\! -\! \alpha^2\!\left(\!1\!+\!\frac{1}{\gamma_{i,u}}\!\right)\!\lambda_{i,u}\bh_{i,i,u}\bh_{i,i,u}^H \label{eq:derivative} \\
			&+\alpha^2\!\sum_{j,v}\lambda_{j,v}\bh_{i,j,v}\bh_{i,j,v}^H \nonumber 
			+\!\alpha(1-\alpha){\rm diag}\!\left(\bH_i\,{\pmb\Lambda}\,\bH_i^H\!\right)\!\bigg)\bw_{i,u}. \nonumber
		\end{align}
		We then set the derivative to zero, and solve it for $\bw_{i,u}$ as
		\begin{align}
			\nonumber
			&\bw_{i,u} =  \bigg( \!\alpha^2\!\!\!\!\!\sum_{(j,v)\neq (i,u)} \!\!\!\lambda_{j,v} \bh_{i,j,v}\bh_{i,j,v}^H\! \\ 
			&+\! \alpha(1\!-\!\alpha){\rm diag}\!\left(\bH_i{\pmb\Lambda}\bH_i^H\!\right)\! \nonumber +\! \alpha \bD_i\! \bigg)^{\!-1}\!\!\frac{\alpha^2}{\gamma_{i,u}}\!\lambda_{i,u}\bh_{i,i,u}\bh_{i,i,u}^H\bw_{i,u}\nonumber  \\
			\nonumber
			& = \frac{\alpha^2}{\gamma_{i,u}}\lambda_{i,u}\bh_{i,i,u}^H\bw_{i,u} \bff_{i,u}
		\end{align}
		where $\bff_{i,u}$ is in \eqref{eq:mmse} and the last equality is valid because $\bh_{i,i,u}^H\bw_{i,u}$ is a scalar.
		Accordingly, we can justify the form of $\bw_{i,u} = \sqrt{\tau_{i,u}}\bff_{i,u}$ with properly designed $\tau_{i,u}$.
		
		To satisfy the KKT stationarity condition with the DL constraint in \eqref{eq:dl_sinr}, $\Gamma_{i,u}$ has to meet the target SINR constraint with equality.
		Since the DL precoder is deeply embedded in the quantization noise term, i.e., ${\rm Q}_{i,u}$, in the DL SINR expression, we first simplify ${\rm Q}_{i,u}$. 
		Let us define $\mu_{i',u'}$ where $\mu_{i',u'}=1$ if $i'=i$ and $u'=u$, and $\mu_{i',u'}=0$ otherwise. 
		To compose the DL constraint in a tractable form, we can rewrite the quantization error term in \eqref{eq:sinr_dl_ofdm} as
		\begin{align}
			\nonumber
			{\rm Q}_{i,u} &\stackrel{(a)}=\! \alpha\beta\sum_{j=1}^{N_c}\bh_{j,i,u}^H{\rm diag}\big(\ubW_j\ubW_j^H\big)\bh_{j,i,u}\! \\
			\nonumber
			&=\alpha\beta\sum_{i',u',j}\mu_{i',u'}\bh_{j,i',u'}^H
			{\rm diag}\big(\ubW_j\ubW_j^H\big)\bh_{j,i',u'}
			\\ \nonumber
			&=\alpha\beta\sum_{j,v}\bw_{j,v}^H
			{\rm diag}\left(\bH_j\,\ubM\,\bH_j^H\right)\!\bw_{j,v}
			\\ 
			&\stackrel{(b)}= \alpha\beta\sum_{j,v}\bw_{j,v}^H
			 {\rm diag}\left(\bh_{j,i,u}\bh_{j,i,u}^H\!^H\!\right)\!\bw_{j,v}^H, 
			\label{eq:QN_reform_ofdm}
		\end{align}
		% changing the indices {\color{red}  from $(j,v,\ell,i')$ back to $(i',u',n,j)$ (JC: no need to change?)} gives us  by
		where $(a)$ is obtained by plugging \eqref{eq:Cqq_dl} into \eqref{eq:quantizationnoise}.
		We define $\bM_i \!=\! {\rm diag}(\mu_{i,1},\dots,\mu_{i,N_u})$ and ${\bM} \!=\! {\rm blkdiag}({\bM}_1,\dots, {\bM}_{N_c})$.
		By the definition of $\bM$, $(b)$ is from $\bH_j\,\ubM\,\bH_j^H = \bh_{j,i,u}\bh_{j,i,u}^H$.
		% 's activates only $(kN_u+u)$th column of $\ubG_{j,i}$ thereby satisfying $\ubG_j\,\ubM\,\ubG_j^H = \ubg_{j,i,u}\!\ubg_{j,i,u}^H$ in .
	
Accordingly, we can create the DL SINR constraints with equality as follows:
\begin{align}
	\nonumber
	\sigma^2&=\frac{\alpha^2}{\gamma_{i,u}} |\bh_{i,i,u}^H{\bf w}_{i,u}|^2 -\alpha^2 \!\!\!\!\sum_{(j,v) \neq (i,u)}^{N_c, N_u}\!|\bh_{j,i,u}^H\bw_{j,v}|^2  -{\rm Q}_{i,u} \\  
	\nonumber
	&\stackrel{(a)}=  \frac{\alpha^2}{\gamma_{i,u}} |\bh_{i,i,u}^H{\bf f}_{i,u}|^2 \tau_{i,u} 
	\nonumber
	-  \alpha^2 \!\!\!\!\!\!\!\sum_{(j,v) \neq (i,u)}\!\!\!\!\! {| \bh_{j,i,u} \bff_{j,v}^H |^2} \tau_{j,v} 
	\!\\
	&\quad-\! \alpha\beta\!\sum_{j,v}\!\bw_{j,v}^H{\rm diag}\!\left(\bh_{j,i,u}\bh_{j,i,u}^H\!^H\right)\!\bw_{j,v}^H, \label{eq:dl_active}
	% or the quantization error term can be expressed as \sum_{n}^{N_b}{h_{j,i,u,n}^2 f_{i,u,n}^2}
	%&= 1, \quad \forall i, u,
\end{align} 
for all $i$, $u$ where $(a)$ is from \eqref{eq:QN_reform_ofdm} and $\bw_{i,u} = \sqrt{\tau_{i,u}}\bff_{i,u}$. 
% Noting that $\bff_{i,u}$ and $\bh_{j,i,u}$ are previously known, 
Combining \eqref{eq:dl_active} for all $i,u$ gives a system of linear equations as $\sigma^2{\bf 1} = {\bSigma}{\btau} $, thereby having $\tau_{i,u}$ through ${\btau}=\sigma^2 {\bSigma}^{-1}{\bf 1}$.
% {\color{black}To show that the derived $\bw_{i,u}$'s are the optimal solution, it suffices to show that KKT conditions hold because the primal downlink can be represented as a convex problem with a differentiable objective function. By solving \eqref{eq:dl_active} with the matrix inversion, $\bw_{i,u}$'s satisfies the downlink constraint in \eqref{eq:problem_dl} with equality conditions, consequently satisfying both primal feasibility and complementary slackness condition. Following the proof of the Corollary~\ref{cor:solution}, the optimal $\lambda_{i,u}$'s fulfill the dual feasibility as well as stationary condition.}
\end{proof}
\end{corollary}

%%%%%%%%%%%%%%%%%%%%%%%%%%%%%%%%%%%
\subsection{Distributed Iterative Algorithm}
\label{subsec:algorithm}
%%%%%%%%%%%%%%%%%%%%%%%%%%%%%%%%%%%

In this subsection, we characterize solutions by exploiting strong duality and further adopt the fixed-point iteration with a subgradient projection method \cite{dahrouj2010coordinated}.

\begin{corollary}
\label{cor:solution_ofdm}
For a fixed $\bD_i$, the optimal power for the uplink total
power minimization problem in \eqref{eq:dual_problem} is given as
\begin{align}
    \label{eq:solution_ofdm}
    \lambda_{i,u} = \frac{1}{\alpha \left(1+\frac{1}{\gamma_{i,u}}\right)\bh_{i,i,u}^H \bK_{i}^{-1}({\pmb\Lambda}) \bh_{i,i,u}},
\end{align}
where $\bK_{i}({\pmb\Lambda})$ is defined in \eqref{eq:K_matrix}.
\end{corollary}
\begin{proof}
Setting \eqref{eq:derivative} to zero and solving it for $\lambda_{i,u}$ produce \eqref{eq:solution_ofdm}. The solution of $\lambda_{i,u}$ then satisfies the stationary condition, and we further observe that the UL SINR constraint in \eqref{eq:per_antenna_const}
is active at the solution satisfying the complementary slackness condition. Therefore, \eqref{eq:solution_ofdm} is optimal solution of the virtual UL problem.
\end{proof}

The UL power control solution from Corollary~\ref{cor:solution_ofdm} is fundamentally designed for the subproblem of \eqref{eq:dual_problem} written as
\begin{align}
        \label{eq:subproblem}
        f\!\left(\bD_i\right)=&\min_{\lambda_{i,u}} \;\;  \sum_{i,u}^{N_c,N_u}\lambda_{i,u}\sigma^2\\
        \nonumber
        &{\text{ \rm subject to}}  \;\; \max_{\bff_{i,u}}\hat{\Gamma}_{i,u} \geq \gamma_{i,u} \;\;\forall i,u,
\end{align}
which is the inner optimization on $\lambda_{i,u}$ of $g(\bD_i, \lambda_{i,u})$ when MMSE filter in \eqref{eq:mmse} is used for $\bff_{i,u}$. 
However, the solutions do not guarantee a global optimum over the entire feasible candidates of $\bD_i$ in \eqref{eq:D_const1}-\eqref{eq:D_const2}. Therefore, as a second stage, the external loop on $\bD_i$ is needed. 
We then use a projected subgradient ascend method to maximize the objective function while satisfying the constraint on $\bD_i$. Note that $f\!\left(\bD_i\right)$ is concave in $\bD_i$ since the Lagrangian function in $\eqref{eq:lagrangian}$ is affine in $\lambda$ and $\mu$, and the infimum of affine functions is still concave. 

\begin{corollary}[Subgradient]
	 ${\rm diag}\left(\sum_{u}\bw_{i,u}\bw_{i,u}^H\right)$ is a subgradient of \eqref{eq:subproblem} in updating $\bD_i$.
	\begin{proof}
	Using strong duality between UL and DL, the UL subproblem $f\!\left(\bD_i\right)$ in \eqref{eq:subproblem} for a fixed $\bD_i$ is equivalent to
	\begin{align}
		\label{eq:subproblem1}
		f\!\left(\bD_i\right)=&\min_{\bw_{i,u}} \;\;  \sum_{u}^{N_u}\bw_{i,u}^H\bD_i\bw_{i,u}\\
		\nonumber
		&{\text{ \rm subject to}}  \;\; \Gamma_{i,u} \geq \gamma_{i,u} \;\;\forall i,u,
	\end{align}
	based on the proof of Theorem \ref{thm:duality_ofdm}.
	We introduce two arbitrary diagonal covariance matrices $\bD_i$ and $\bD_i^\prime$ whose associated optimal beamformer is $\bw_{i,u}$ and $\bw_{i,u}^\prime$, respectively.
	We then derive a subgradient as follows:
	\begin{align}
		&f\!\left(\bD_i^\prime\right)-f\!\left(\bD_i\right)=\sum_{u}\bw_{i,u}^{'H}\bD_i^\prime\bw_{i,u}^\prime-\sum_{u}\bw_{i,u}^H\bD_i\bw_{i,u} \\
		&\qquad\qquad\qquad\;\;\leq\sum_{u}\bw_{i,u}^H\bD_i^\prime\bw_{i,u}-\sum_{u}\bw_{i,u}^H\bD_i\bw_{i,u} \\
		&\qquad \qquad={\rm tr}\left({\rm diag}\left(\sum_{u}\bw_{i,u}\bw_{i,u}^H\right)\left(\bD_i^\prime-\bD_i\right)\right).
	\end{align}
	By the definition of subgradients, the multiplier in front of $\bD_i^\prime-\bD_i$ becomes a subgradient of $f(\bD_i)$.
    \end{proof}
\end{corollary}
We finally plug the subproblem into the entire DL problem in \eqref{eq:dual_problem}. 
Recall that the main target is to maximize $f(\bD_i)$ while satisfying the constraints on $\bD_i$. 
After obtaining a converged solution of \eqref{eq:subproblem} for a fixed $\bD_i$, we take a step in the direction of a positive subgradient. We further project the updated $\bD_i$ onto the feasible set on $\bD_i$ in \eqref{eq:D_const1} and \eqref{eq:D_const2} because the updated $\bD_i$ probably violates the feasible domain.
The complete algorithm is summarized in Algorithm 1.

\begin{algorithm}[t]
	\caption{Joint Power-Minimizing Transmission with Per-Antenna Constraints
% 	Quantization-aware iterative CoMP with per-antenna constraints (Q-iCoMP-PA)
	}
% 	\begin{algorithmic}[1]
%		\renewcommand{\algorithmicrequire}{\textbf{Input:}}
%		\renewcommand{\algorithmicensure}{\textbf{Output:}}
%		\REQUIRE in
%		\ENSURE  out
		 Initialize $\lambda_{i,u}$, $\forall i,u$ and $\bD_{i}^{(0)}$, $\forall i$.\\
		\While{$\bD_{i}^{(n)}$'s do not converge}{
		Iteratively update $\lambda_{i,u}$ until converges  as
		\begin{align}
			\nonumber
			\lambda_{i,u} = 
			 \frac{1}{\alpha \Big(1+\frac{1}{\gamma_{i,u}}\Big)\bh_{i,i,u}^H \left[\bK_{i}^{(n)}({\pmb\Lambda})\right]^{-1} \bh_{i,i,u}}, \nonumber
		\end{align}
		for all $i,u$ where $\bK_{i}^{(n)}({\pmb\Lambda})$ is updated according to \eqref{eq:K_matrix} using $\lambda_{i,u}$ and $\bD_i^{(n)}$.\\  
		Find the UL MMSE equalizer ${\bf f}_{i,u}$ in \eqref{eq:mmse} with $\bD_i^{(n)}$ and the converged $\lambda_{i,u}$\\
		Compute the DL precoder $\bw_{i,u}$ from Corollary~\ref{cor:dl_precoder_ofdm}.\\
		\For{$i=1$ to $N_c$}{
		$\bD_{i}^{(n+1)} \!\gets\!  \bD_{i}^{(n)} \!+\! \eta\; {\rm diag}\left(\sum_{u}\bw_{i,u}\bw_{i,u}^H\right)$.\\
		Project $\bD_{i}^{(n+1)}$ onto the feasible set \eqref{eq:D_const1}-\eqref{eq:D_const2} until converges as\\
		$\bD_{i}^{(n+1)} \!\gets\!  \bD_{i}^{(n+1)} \!\!-\! \frac{\rm max \left(0,{\rm tr}\left(\bD_{i}^{(n+1)}\right)-N_b\right)}{\left\|{\bf 1}_{N_b}\right\|^2}{\bf 1}_{N_b}$}
    	$n \gets n+1$
		}
%		\FOR {$i = l-2$ to $0$}
%		\STATE statements..
%		\IF {($i \ne 0$)}
%		\STATE statement..
%		\ENDIF
%		\ENDFOR
%		\RETURN $P$
% 	\end{algorithmic}
 \Return{\ }$\bw_{i,u}$ for all $i,u$.
\end{algorithm}

\section{Simulation Results}
%%%%%%%%%%%%%%%%%%%%%%%%%%%%%%%%%%

We evaluate the derived results of the proposed
quantization-aware iterative CoMP algorithm with per-antenna constraints (Q-iCoMP-PA) against
the quantization-aware iterative CoMP algorithm (Q-iCoMP) in \cite{choi2020quantized}. 
The former is based on the antenna power minimax problem while the latter focuses on the total transmit power minimization problem.

Each BS is in the center of own hexagonal cell and BSs operate beside each other.
We assume that the small scale fading of each channel follows Rayleigh fading with zero mean and unit variance. 
For large scale fading, we use the log-distance pathloss model in \cite{erceg1999empirically}. 
The distance between adjacent BSs is $2\  \rm km$.
The minimum distance between BS and user is $100\  \rm m$.
Considering a $2.4 \ \rm GHz$ carrier frequency with $10\  \rm MHz$ bandwidth, we use $8.7\  \rm dB$ lognormal shadowing variance and $5 \ \rm dB$ noise figure.
For simplicity, we assume an equal target SINR for all users, i.e., $\gamma_{i,u}=\gamma$ for all $i,u$.
\begin{figure}[!t]\centering
	\includegraphics[width=1.0\columnwidth]{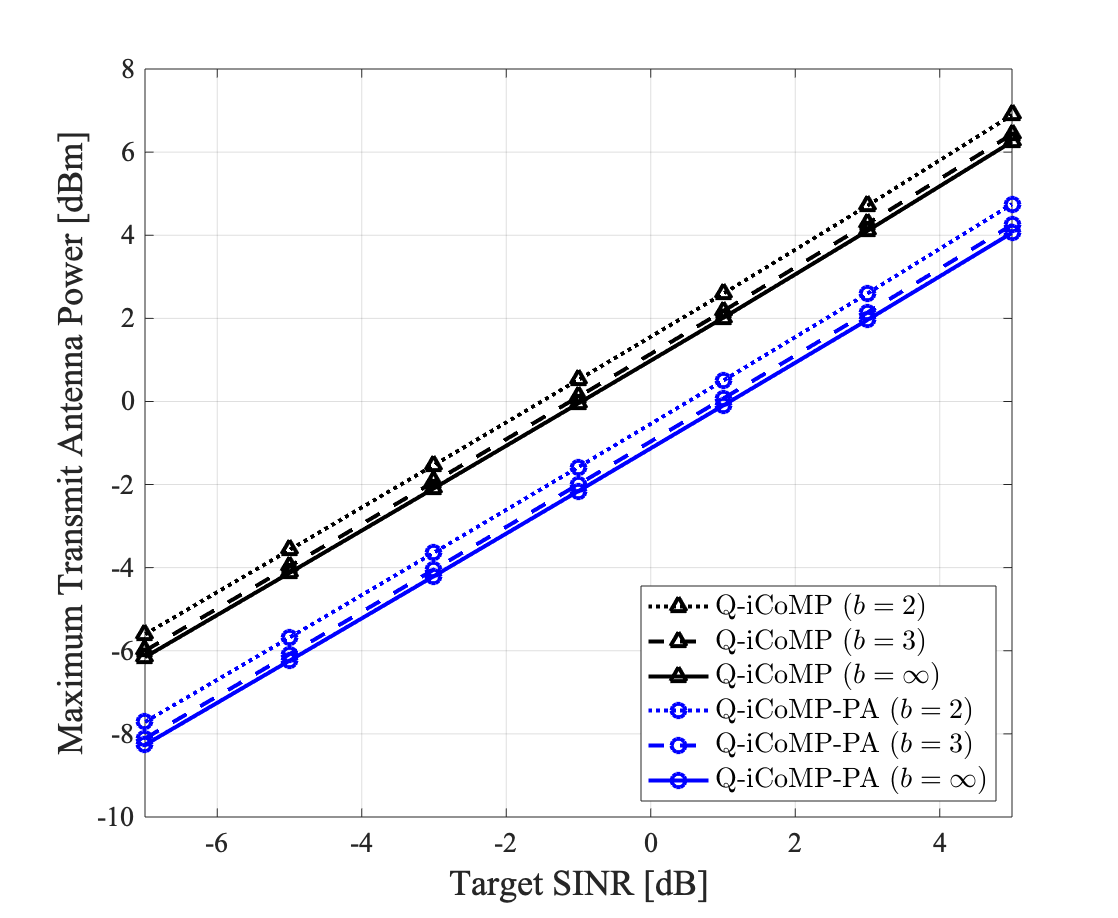}
	\vspace{-1.0em}
	\caption{Maximum transmit antenna power vs. target SINR for the network with $N_b = 32$ antennas per BS, $N_c = 4$ cells,  $N_u = 2$ users per cell, and $b \in \{2,3,\infty\}$ bits using 20 channel realizations per target SINR value.} 
	\label{fig:power}
	\vspace{-1.0em}
\end{figure}

Fig.~\ref{fig:power} shows maximum transmit antenna power in DL direction across all $N_c N_b$ transmit antennas for given target SINRs. 
We consider a communication configuration with $N_b=32$ BS antennas, $N_c = 4$ cells, and $N_u=2$ users per cell.
We test both infinite-resolution and low-resolution converters, i.e., $b\!\in\!\{2,3,\infty\}$.
When using infinite-resolution ADCs and DACs, we have slightly lower peak power compared to one with 3-bit data converters, however the gap between $b=2$, $b=3$, and $b=\infty$ cases is marginal on both Q-iCoMP and Q-iCoMP-PA by properly incorporating the coarse quantization error into the design of beamformers.
With multi-cell coordination, both Q-iCoMP and Q-iCoMP-PA do not suffer from implausible power consumption and undesirable divergence.
However, based on the primal problem of Q-iCoMP-PA, the proposed algorithm can limit the maximum transmit power providing around 2 dB gain over the regular Q-iCoMP. 

\begin{figure}[!t]\centering
	\includegraphics[width=1\columnwidth]{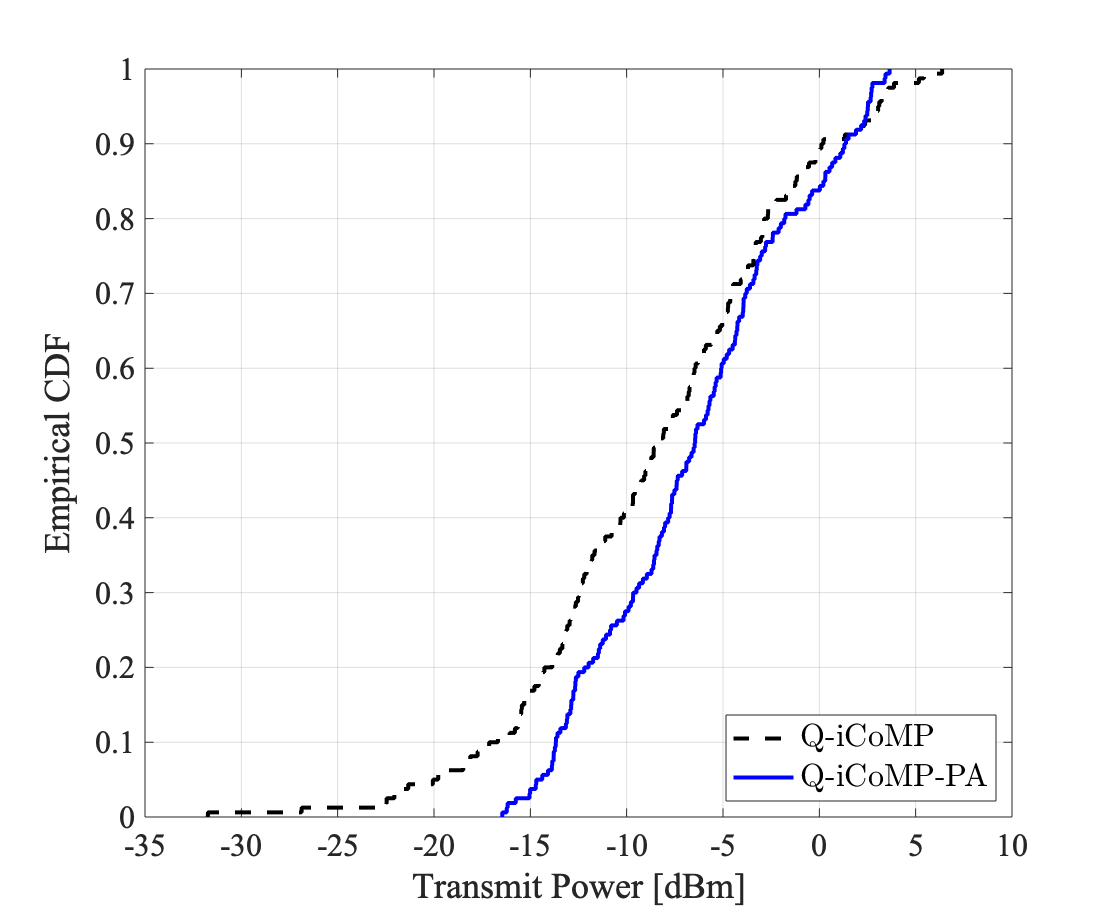}
	\vspace{-1em}
	\caption{
		Empirical CDF of the transmit power of the BS antennas in the network with $N_b = 32$ antennas per BS, $N_c = 5$ cells,  $N_u = 2$ users per cell, $b=3$ quantization bits, and $\gamma=2 \ \rm dB$ target SINR using one channel realization. Empirical CDF gives the probability of an antenna operating at or below the given transmit power.} 
	\label{fig:cdf}
	\vspace{-0.5em}
\end{figure}

Fig.~\ref{fig:cdf} shows the cumulative density function (CDF) of the transmit power of all antennas, i.e., $N_bN_c$ antennas, considering one channel realization. 
We employ a network with $N_b=32$ BS antennas, $N_c = 5$ cells, $N_u=2$ users per cell, and $b=3$ quantization bits.
From the figure, the proposed algorithm reveals two main advantages: 1) maximum transmit antenna power; and 2) operating range.
Since the CDF is plotted over all antennas, the rightmost point of the CDF represents the maximum transmit  power. 
When comparing the rightmost point of two methods, Q-iCoMP-PA achieves more than $2 \ \rm dB$ gain over Q-iCoMP, which corresponds to the main purpose of the proposed method.
Also, the CDF gives the operating range which is defined as the gap between the leftmost and rightmost points of the CDF.
Q-iCoMP-PA works with a much narrower operating range compared with Q-iCoMP, thereby increaseing efficiency of power-related components such as power amplifier.

\begin{table}[!t]
\caption{Comparison of Peak-to-average power ratio (PAPR).}
\begin{center}
\begin{tabular}{|c|c|c|}
\hline
 $\gamma$ & \multicolumn{2}{|c|}{{Methods}} \\
\cline{2-3}
 [dB] &  Q-iCoMP & Q-iCoMP-PA \\
\hline
2 & 3.81 dB & 2.16 dB \\
\hline
-3 & 3.77 dB & 2.09 dB \\
\hline
\end{tabular}
\\
\vspace{0.2cm}
$N_b = 32$ antennas per BS, $N_c = 4$ cells,  $N_u = 2$ users per cell, and $b=3$ bits. \\
\label{table:papr}
\end{center}
\end{table}

In Table.~\ref{table:papr}, we further simulate the peak-to-average power ratio (PAPR) which is directly related to the efficiency of power amplifiers. 
We consider a network configuration with $N_b = 16$ BS antennas, $N_c \in \{2,3\}$ cells,  $N_u = 2$ users per cell, and $b=3$ bits over different constraint SINRs.
%Both Q-iCoMP and Q-iCoMP-PA exhibit a stable PAPR regardless of the target SINR.
With two cells in the network, the Q-iCoMP-PA achieves significant reduction over the Q-iCoMP, showing more than 1.6 {\rm dB} gain on average. 
Therefore, Q-iCoMP-PA is more favorable for mobile communication systems by properly providing multi-cell coordination and limiting the peak power of antennas.

%%%%%%%%%%%%%%%%%%%%%%%%%%%%%%%%%%
\section{Conclusion}
%%%%%%%%%%%%%%%%%%%%%%%%%%%%%%%%%%

In this paper, we investigated the CoMP solution for a multicell configuration with low-resolution data converters when employing per-antenna power constraints for more practical deployment.
Considering the coarse quantization error, we derived the antenna power minimax problem and effective UL problem with uncertain noise covariance as dual problem, and further proved zero duality gap between two problems.
Leveraging strong duality, we proposed the iterative algorithm that finds the optimal dual solution for a fixed covariance and used the solution to compute the optimal DL beamformer.
We further update the UL noise covariance using the optimal DL solution with the projected subgradient descent method.
In simulation, the proposed Q-iCoMP-PA achieves significant gain over Q-iCoMP in terms of maximum antenna power, operating range, and PAPR, thereby improving hardware efficiency.
Therefore, we can emphasize the need to limit the antenna power when deploying multicell massive MIMO communication systems with low-resolution data converters.

\bibliographystyle{IEEEtran}
\bibliography{CoMP_ADCs.bib}

\end{document}